\documentclass[nonacm]{aamas} 

\usepackage{graphicx}  
\usepackage{subcaption}
\usepackage{amsmath}
\usepackage{framed} 

\usepackage{amssymb}
\usepackage{balance} 

\setcopyright{ifaamas}
\acmConference[AAMAS '23]{Proc.\@ of the 22nd International Conference
on Autonomous Agents and Multiagent Systems (AAMAS 2023)}{May 29 -- June 2, 2023}
{London, United Kingdom}{A.~Ricci, W.~Yeoh, N.~Agmon, B.~An (eds.)}
\copyrightyear{2023}
\acmYear{2023}
\acmDOI{}
\acmPrice{}
\acmISBN{}

\acmSubmissionID{384}

\title[AAMAS-2023 Formatting Instructions]{A Redistribution Framework for Diffusion Auctions}

\author{Sizhe Gu}
\affiliation{
  \institution{ShanghaiTech University}
  \city{Shanghai}
  \country{China}}
\email{guszh@shanghaitech.edu.cn}

\author{Yao Zhang}
\affiliation{
  \institution{ShanghaiTech University}
  \city{Shanghai}
  \country{China}}
\email{zhangyao1@shanghaitech.edu.cn}

\author{Yida Zhao}
\affiliation{
  \institution{ShanghaiTech University}
  \city{Shanghai}
  \country{China}}
\email{zhaoyd1@shanghaitech.edu.cn}

\author{Dengji Zhao}
\affiliation{
  \institution{ShanghaiTech University}
  \city{Shanghai}
  \country{China}}
\email{zhaodj@shanghaitech.edu.cn}

\begin{abstract}
Redistribution mechanism design aims to redistribute the revenue collected by a truthful auction back to its participants without affecting the truthfulness. We study redistribution mechanisms for diffusion auctions, which is a new trend in mechanism design~\cite{zhao22mechanism}. The key property of a diffusion auction is that the existing participants are incentivized to invite new participants to join the auctions. Hence, when we design redistributions, we also need to maintain this incentive. Existing redistribution mechanisms in the traditional setting are targeted at modifying the payment design of a truthful mechanism, such as the Vickrey auction. In this paper, we do not focus on one specific mechanism. Instead, we propose a general framework to redistribute the revenue back for all truthful diffusion auctions for selling a single item. The framework treats the original truthful diffusion auction as a black box, and it does not affect its truthfulness. The framework can also distribute back almost all the revenue.

\end{abstract}


\newcommand{\BibTeX}{\rm B\kern-.05em{\sc i\kern-.025em b}\kern-.08em\TeX}

\begin{document}


\pagestyle{fancy}
\fancyhead{}


\maketitle 


\section{Introduction}
We focus on a resource allocation problem that involves a group of self-interested agents competing for the resources. One important goal of the allocation problem is to maximize social welfare, and a common method is to hold an auction so that the agents with the highest valuations of the resources can be found. The Vickrey-Clarke-Groves (VCG) mechanism~\cite{vickrey1961counterspeculation,clarke1971multipart,groves1973incentives} is a well-known method under which the agents will truthfully report their valuations and the resources will be allocated to the agents with the highest valuations. This will also lead to a high revenue for the seller. However, in many situations, we do not seek for a profit. Hence, there is another body of studies on how to return the revenue to the participants~\cite{DBLP:conf/atal/Cavallo06,guo2009worst,DBLP:conf/aaai/Guo11,guo2016competitive,guo2011vcg,moulin2009almost,manisha2018learning}, which is called redistribution mechanism design. 

In recent years, researchers have started to design mechanisms on social networks~\cite{jackson2010social,DBLP:books/daglib/0025903}, where the goal is to incentivize agents to invite new agents via their social connections to join in~\cite{DBLP:conf/aaai/LiHZZ17,DBLP:conf/atal/ZhaoLXHJ18,DBLP:conf/ijcai/LiHZY19,zhang2020incentivize,DBLP:conf/aaai/KawasakiBTTY20,zhao2021mechanism}. By doing so, the mechanism can further improve social welfare or revenue in, for example, auctions. 
Therefore, introducing social network to redistribution problems is a good choice to break through the efficiency limitations of traditional settings.
There are many realistic scenarios of redistribution on social networks. For example, consider a non-profit organization like a government who has some idle resources, such as properties confiscated from criminals. These resources are normally destroyed and sold cheaply to a small group of people, as the organization does not want to gain much profit from clearing the resources. By utilizing  social networks, we could attract more buyers who are willing to pay more to receive the idle resources, but the profit could still be redistributed back to participants. 

Redistribution on networks is also more challenging than traditional settings because the action space of the participants is enlarged. Therefore, we are not able to directly apply the existing solutions such as Cavallo's mechanism~\cite{DBLP:conf/atal/Cavallo06}. The intuition behind this failure is that agents will have no incentives to invite their competitors in an auction, or invite others to share a limited redistribution. Another difficulty is that agents' valuations not only determine the winner of the auction, but also relate to how much revenue can be redistributed. Therefore, the redistribution problem is not about a simple combination of a truthful auction mechanism and a truthful reward distribution mechanism, which may not produce a truthful redistribution mechanism. 

More importantly, there are a bunch of different diffusion mechanisms with different allocation and payment policies. 
It is complex and tedious to design the redistribution mechanism for each diffusion auction separately. Hence, we design the first general redistribution framework for all diffusion auctions. The framework can redistribute almost all the revenue of any diffusion auction back to all agents without affecting the properties of incentive compatibility and individual rationality. 
In fact, our framework can also be applied to all traditional auctions. In particular, when the input mechanism is VCG, the mechanism generated by our framework is Cavallo's mechanism. Therefore, our framework is a general solution for redistribution mechanism design with or without networks.


\section{Preliminaries}
\par We consider a setting where a sponsor $s$ wants to allocate a single item on a social network. Apart from the sponsor, the social network consists of $n$ agents denoted by $N = \{1,...,n\}$. Each agent $i \in N \cup \{ s \}$ has a private neighbour set $r_i \subseteq N \cup \{ s \} \backslash \{i\}$, which represents the agents with whom $i$ can communicate directly.
Furthermore, each agent $i \in N$ has a private valuation for the item of $v_i \geq 0$. 
Initially, the sponsor can only invite her neighbours to participate in the allocation. In order to attract more participants, the sponsor asks participants who have already joined the allocation to further invite their neighbours to join. However, participants are competitors, so they would not invite each other by default. Thus, we need to design mechanisms to incentivize them to invite each other, which are called diffusion mechanisms. 

Formally, the diffusion mechanism requires each agent to report not only her valuation but also her neighbour set (which is equivalent to inviting her neighbours). Let $\theta_i = (v_i, r_i)$ be the type of an agent $i\in N$ and $\theta_i' = (v_i', r_i')$ be the reported type of $i$, with $r_i' \subseteq r_i$ representing the actual neighbour set she has invited, which can only be a subset of her true neighbour set $r_i$. Let $\theta' = (\theta'_1, \dots , \theta'_n) = (\theta'_i, \theta'_{-i})$ be the overall report profile, and $\theta'_{-i}$ is the overall report profile except for $i$. Let $\Theta_i$ represent the type space of agent $i$ and $\Theta = (\Theta_1, \dots, \Theta_n)$ represent the type profile space of all agents. Given any report profile $\theta'$, it induces a directed graph denoted by $G(\theta') = (V(\theta'), E(\theta')) $, where $V(\theta') = N \cup \{s\}$ and $E(\theta')= \{(s,j)|j\in r_s\} \cup \{(i,j)| i\in N, j\in r_i'\}$. Let $D_s(G(\theta'))$ denote the set of agents accessible from $s$ in $G(\theta')$. Because the others cannot receive the proper invitation started by the sponsor, only the agents in $D_s(G(\theta'))$ can actually join the allocation (in practice, this means that the others will not be informed about the allocation at all). An example of the induced graph $G(\theta')$ is shown in Figure~\ref{fig:induced}.

The goal of the sponsor is to allocate the item to the agent with the highest valuation, but she does not want to gain any profit from the allocation. Hence, the sponsor need to redistribute the revenue from a given diffusion auction mechanism. 

\begin{definition}\label{def:mechanism}

A diffusion auction mechanism $\mathcal{M}$ is defined by an allocation policy $\pi = \{\pi_1,..,\pi_n\}$ and a payment policy $x = \{x_1,...,x_n\}$, where $\pi_i: \Theta \rightarrow \{0,1\}$ and $x_i : \Theta \rightarrow \mathbb{R}$ are the allocation and payment for $i$ respectively. Additionally, for all report profiles $\theta'\in \Theta$,

\begin{itemize}
    \item for any agent $i\not\in D_s(G(\theta'))$, $\pi_i(\theta') = 0$ and $x_i(\theta') = 0$;
    \item for any agent $i\in D_s(G(\theta'))$, $\pi_i(\theta')$ and $x_i(\theta')$ are independent of the reports of agents who are not in $D_s(G(\theta'))$;
    \item $\sum_{i \in N}\pi_i(\theta') \leq 1$. 
\end{itemize}
\end{definition}

\begin{figure}[htbp]
    \centering  
	\begin{subfigure}[b]{0.275\textwidth}
        \includegraphics[width=\textwidth]{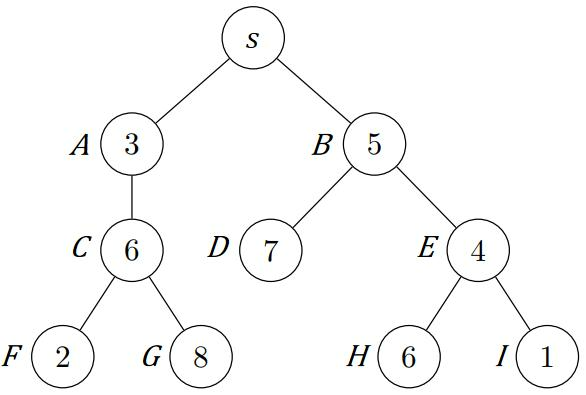}
        \caption*{(1)}
        \label{fig:2.11}
    \end{subfigure}
    \begin{subfigure}[b]{0.175\textwidth}
        \includegraphics[width=\textwidth]{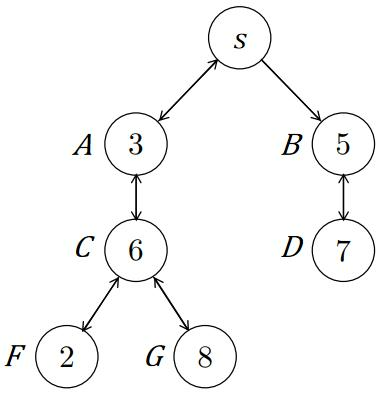}
        \caption*{(2)}
        \label{fig:2.12}
    \end{subfigure} 
    \caption{An example of induced graphs by agents' type profile and report profile. The number in each node is the real/reported valuation of the agent. (1) is the graph induced by agents' real types. (2) is the graph induced by agents' report profile where agent $B$ misreports $\theta_B' = (5, \{D\})$ (only agents in $D_s(G(\theta'))$ are shown). Then, agents $E$, $H$ and $I$ cannot participate in the mechanism.}
    \label{fig:induced}
\end{figure}

Given the agents’ report profile $\theta'$, $\pi_i(\theta') = 1$ means that the item is allocated to agent $i$, while $\pi_i(\theta') = 0$ means that 
$i$ does not get the item. In addition, $x_i(\theta') \geq 0$ means that 
$i$ pays $x_i(\theta')$ to the sponsor, and $x_i(\theta') < 0$ indicates that 
$i$ receives $|x_i(\theta')|$ from the sponsor. Therefore, the surplus of the payment transfers in the mechanism (i.e., the revenue) is defined by
$$ S(\theta') = \sum_{i \in N}x_i(\theta').$$
\par Then in our setting, given the type $\theta_i = (v_i , r_i)$ of an agent $i$ and the report profile $\theta'$, the utility of agent $i$ under the diffusion auction mechanism $\mathcal{M} = (\pi, x)$ is
$$ u_i(\theta_i, \theta') = \pi_i(\theta') \cdot v_i - x_i(\theta') .$$

In the following, we will define several properties that are required for the diffusion auction mechanism. First, an agent should not suffer a loss if she reports her true valuation on the item. 

\begin{definition}
    A diffusion auction mechanism $\mathcal{M}=(\pi, x)$ is \textbf{individually rational} (IR) if for all $i \in N$ and all $\theta' \in \Theta$, we have $$u_i(\theta_i,((v_i,r_i'),\theta_{-i}'))\geq 0.$$
\end{definition}

Next, we want to incentivize all agents not only to report their true valuations but also to invite all their neighbours to join the mechanism, i.e., reporting their true types is a dominant strategy.

\begin{definition}
    A diffusion auction mechanism $\mathcal{M}=(\pi, x)$ is \textbf{incentive compatible} (IC) if for all $i\in N$ with $\theta_i\in \Theta_i$, and all $\theta_i'\in \Theta_i$, $\theta_{-i}'\in \Theta_{-i}$, we have $$u_i(\theta_i, (\theta_i,\theta_{-i}')) \geq u_i(\theta_i, (\theta_i',\theta_{-i}')).$$
\end{definition}

\par The sum of the payment transfers should be non-negative for the sponsor $s$; otherwise, she will pay a deficit.

\begin{definition}
    A diffusion auction mechanism $\mathcal{M}=(\pi, x)$ is \textbf{non-deficit} (ND) if for all $i \in N$ and all $\theta' \in \Theta$, we have
    $$S(\theta') = \sum_{i \in N} x_i(\theta') \geq 0.$$
\end{definition}
\par Finally, recall that our sponsor does not want to gain any profit. 
Our goal is to establish a framework that can create a diffusion redistribution mechanism from a given diffusion auction mechanism $\mathcal{M}^a=(\pi^a, x^a)$. 
Given the report profile $\theta' \in \Theta$, denote the revenue achieved by $\mathcal{M}^a$ as $S^a(\theta') = \sum_{i\in N} x^a_i(\theta')$. Then our framework is to decide the amount $R_i$ that is returned to each agent $i$, i.e., to decide a redistribution mechanism $\mathcal{M}=(\pi, x)$ with $\pi_i(\theta') = \pi^a_i(\theta')$ and $x_i(\theta') = x^a_i(\theta') - R_i(\theta')$. A diffusion redistribution mechanism, speaking formally, is also a diffusion auction mechanism, but 
we want it to redistribute the revenue back to agents as much as possible, i.e., the sum of the payment transfers should be close to 0, which is define as follows.



\begin{definition}
    A diffusion redistribution mechanism $\mathcal{M}=(\pi, x)$ is \textbf{asymptotically budget-balanced} (ABB) if for all $\theta' \in \Theta$, 
    $$\lim_{n \to \infty} S(\theta') = 0.$$
\end{definition}
It means that the remaining revenue that has not been redistributed is approaching zero if the number of agents in the underlying network is large enough. We also consider an approximation to ABB when it is hard to achieve.

\begin{definition}
    A diffusion redistribution mechanism $\mathcal{M}=(\pi, x)$ is \textbf{$\epsilon$-asymptotically budget-balanced} ($\epsilon$-ABB) if for all $\theta' \in \Theta$,
    $$\lim_{n \to \infty} S(\theta') \leq \epsilon$$ where $\epsilon >0$ is  a constant.
\end{definition}

 
Before we introduce our framework, we show that directly applying classic Cavallo's mechanism fails to satisfy our properties. We recall Cavallo's mechanism as follows.

\begin{framed} 
\noindent\textbf{Cavallo's Mechanism}
\begin{enumerate}
    \item Given $G(\theta')$, the mechanism chooses the winner $w \in \arg\max_{i\in N} v_i'$, and set $\pi_w(\theta') =1$. 
    \item For each agent $i$, her payment is $x_i(\theta') = x_i^{VCG}(\theta') - x_{w'}^{VCG}(\theta_{-i}')/n$, where $x_i^{VCG}(\theta')$ is $i$'s payment under VCG and $w'$ is the winner when we run VCG without agent $i$.
\end{enumerate}
\end{framed}

\begin{proposition}
Cavallo's mechanism is not an incentive compatible diffusion redistribution mechanism.
\end{proposition}

\begin{proof}

\par We prove this by giving an example in Figure~\ref{fig:antiCavallo}. 
First for the situation without diffusion shown in Figure~\ref{fig:antiCavallo}(1), by Cavallo's mechanism, agent $C$ wins the item and $\pi_C(\theta')\cdot v_C= 4$, $x_C^{VCG}(\theta') = 3$. If agent $C$ is removed, $B$ will win the item, so $x_{w'}^{VCG}(\theta_{-C}') = 2$. Similarly we can get $x_{w'}^{VCG}(\theta_{-A}') = 3$ and $x_{w'}^{VCG}(\theta_{-B}') = 2$. 
Therefore, agent $A$, $B$ and $C$ will receive $1$, $\frac{2}{3}$, and $\frac{2}{3}$ as redistribution respectively. However, in Figure~\ref{fig:antiCavallo}(2), when agent $C$ invites agent $D$ and we still apply Cavallo's mechanism, $C$ is still the winner and all VCG payments will not change. Then, agent $C$ can only get $\frac{2}{4} = \frac{1}{2}$ (since $n$ becomes to 4) as redistribution. The payments for getting the item are the same, but the revenue from redistribution decreases after inviting $D$ so agent $C$ has no motivation to invite $D$. 

Hence, Cavallo's mechanism is not IC with diffusion.
\end{proof}

\begin{figure}[htbp]
    \centering  
	\begin{subfigure}[b]{0.18\textwidth}
        \includegraphics[width=\textwidth]{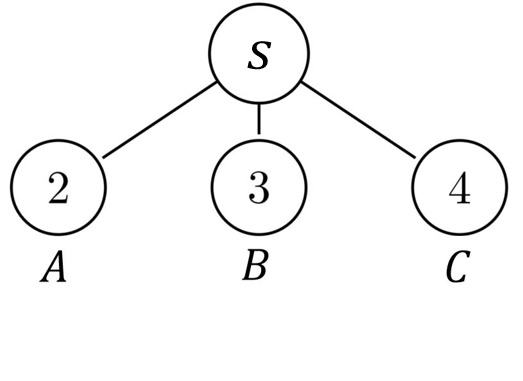}
        \caption*{(1)}
        \label{fig:2.11}
    \end{subfigure}
     \begin{subfigure}[b]{0.045\textwidth}
        \includegraphics[width=\textwidth]{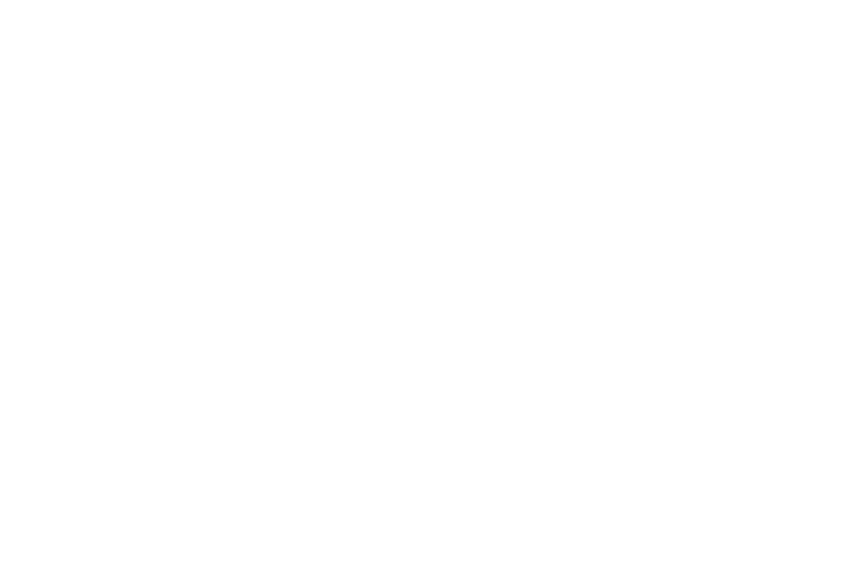}
        \caption*{}
    \end{subfigure}
    \begin{subfigure}[b]{0.18\textwidth}
        \includegraphics[width=\textwidth]{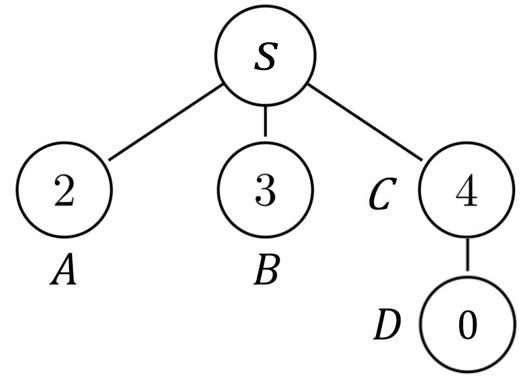}
        \caption*{(2)}
        \label{fig:2.12}
        \end{subfigure}
    \caption{An example of Cavallo’s mechanism being applied in social network. (1) is a graph of traditional setting of redistribution problem. (2) is a graph induced by agents’ real types on social networks. (1) can also be viewed as the result that agent $C$ does not invite agent $D$ in (2). }
    \label{fig:antiCavallo}
\end{figure}

\section {Redistribution Framework for Diffusion Auctions}
\par In this section, we propose a network-based redistribution mechanism framework (NRMF) for all IC diffusion auction mechanisms. 
We first consider a subproblem to redistribute a value to all agents without affect the IC property.

\subsection{Proportional Reward Sharing in a Tree}
\par We consider a problem that a sponsor $s$ wants to share a reward of $B \geq 0$ with a set of agents $N$ connected as a tree rooted by $s$. 

For a tree $T(N\cup \{s\},E)$ rooted by $s$, we introduce the following two notations.

\begin{itemize}
    \item Let $p_i \in N\cup \{s\}$ be the parent agent of $i\in N$ in $T$. 
    \item Let $C_i$ be the set of all agents in the subtree rooted by $i$ in $T$ (excluding $i$). 
    
\end{itemize}

We define the following distribution mechanism.

\begin{framed}
 \noindent\textbf{Proportional Reward Sharing in a Tree (PRST)}
 
 \noindent\rule{\textwidth}{0.5pt}
 
 \noindent\textsc{Input}: a tree $T=(N\cup \{s\},E)$ and a reward $\mathcal{B}$.
 
 \noindent\rule{\textwidth}{0.5pt}
 
 \begin{enumerate}
    \item Set $\Omega_{s} = 1$.
    \item For each agent $i\in N$, let
    $$ \mathsf{total} = \frac{|C_i| + 1}{|C_{p_i}|}, \quad \mathsf{base} = \frac{1}{|C_{p_i}| - |C_i|} $$
    and then recursively define
    
    $$\left\{
    \begin{array}{l}
    \omega_{i} = \Omega_{p_i} \cdot \left( \mathsf{base} + \left( \mathsf{total} - \mathsf{base} \right) \alpha \right)    \\
    \\
    \Omega_i = \Omega_{p_i} \cdot \mathsf{total} - \omega_i = \Omega_{p_i} \cdot \left( \mathsf{total} - \mathsf{base} \right) \cdot \left( 1 - \alpha \right)
    \end{array}\right.$$
    with predefined parameter $0<\alpha<1$.
    \item Set $b_i = \omega_i  \mathcal{B}$.
 \end{enumerate}
 
 \noindent\rule{\textwidth}{0.5pt}
 
 \noindent\textsc{Output}: the share $b_i$ for all agents $i\in N$.
\end{framed}


In the PRST, $\omega_i$ is the proportion of the reward allocated to agent $i$ and $\Omega_i$ is the proportion of the reward that agent $i$ gives to her descendants. The proportion $\omega_i$ has two parts.  The first part is the basic reward, which is determined by the number of descendants of her siblings ($|C_{p_i}|-|C_i|$ in $\mathsf{base}$). The second part is the reward for her diffusion, which increases proportionally to the number of her descendants ($|C_i|$). 
A running example of the PRST is given in Figure~\ref{fig:rho}.

\begin{figure}[htbp]
    \centering
    \includegraphics[width=0.48\textwidth]{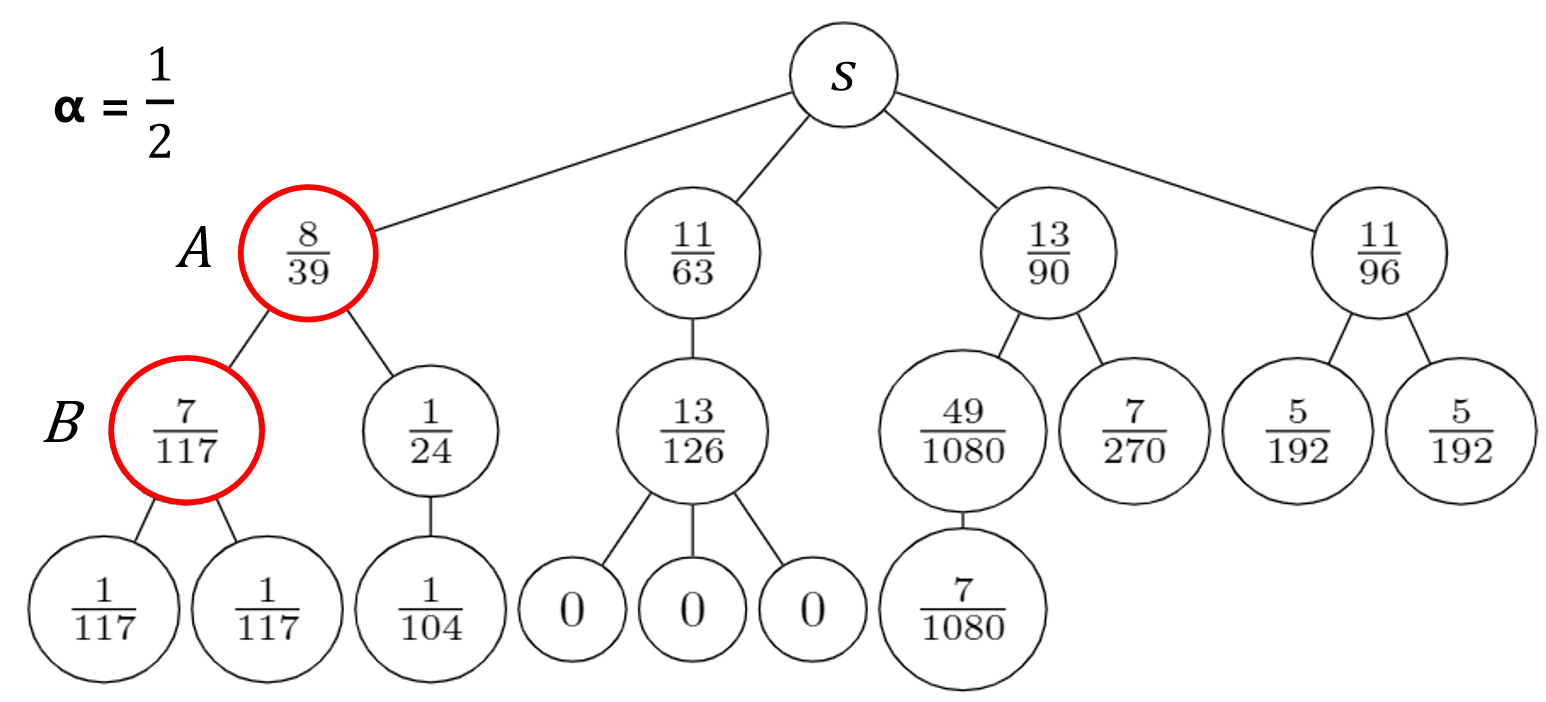}
    \caption{An example of PRST. Here, we set $\mathcal{B} = 1$ and $\alpha = 1/2$. Each node represents an agent and the value in the node is the share $b_i$ of the agent. 
    }
    \label{fig:rho}
\end{figure}

\par In Figure~\ref{fig:rho}, considering agent $A$, $|C_{p_A}|$ means the number of agents in the whole tree excluding $s$, which is 18. $|C_A|$ means the number of descendants of agent $A$, which is 5. According to the definition, $\Omega_{p_A} = \Omega_s = 1$. Therefore,
\begin{align*}
    \omega_A 
    & = \frac{1}{|C_{p_A}|-|C_A|} + \left( \frac{|C_A|+1}{|C_{p_A}|} - \frac{1}{|C_{p_A}|-|C_A|} \right) \cdot \alpha \\
    &= \frac{1}{18-5} + \left( \frac{5+1}{18} - \frac{1}{18-5} \right) \cdot \frac{1}{2}=\frac{8}{39}
\end{align*}
Similarly, considering agent $B$, $|C_{p_B}|$ means the number of descendants of agent $A$, which is 5. $|C_B|$ means the number of descendants of agent $B$, which is 2. Hence, we have 
$\Omega_{p_B} = \Omega_A=\frac{|C_A|+1}{|C_{p_A}|} - \omega_A=\frac{6}{18}-\frac{8}{39} 
$. Therefore,
\begin{align*}
    \omega_B 
    &= \Omega_{p_B} \left[ \frac{1}{|C_{p_B}|-|C_B|} + \left(
        \frac{|C_B|+1}{|C_{p_B}|} - \frac{1}{|C_{p_B}|-|C_B|}\right) \alpha \right]\\
    &=\left( \frac{6}{18}-\frac{8}{39} \right) \cdot \left( \frac{1}{3} + \left( \frac{3}{5} - \frac{1}{3} \right) \cdot \frac{1}{2} \right)=\frac{7}{117}
\end{align*}

In the PRST, each agent can get a basic share when she joins the mechanism, and the base part only depends on the number of agents in the subtrees leading by her siblings. Then she can get a bigger share if she invites her neighbours to join in. The proportion of the share is determined by how many descendants she has compared to her siblings. Intuitively, we introduce a propagation competition among siblings, and the one with a larger propagation will get a larger share of the reward. We illustrate several useful properties of the procedure below.

\begin{lemma}\label{lem:ir}
In the PRST, $b_i\geq 0$ for all $i\in N$.
\end{lemma}

\begin{proof}
Considering $\omega_i$ for all $i\in N$,
$$
    \omega_i = \Omega_{p_i} \left[ \frac{1}{|C_{p_i}|-|C_i|} + \left(
        \frac{|C_i|+1}{|C_{p_i}|} - \frac{1}{|C_{p_i}|-|C_i|} \right) \alpha \right]
$$
Since
$$\frac{1}{|C_{p_i}|-|C_i|} \geq 0, 
    \frac{|C_i|+1}{|C_{p_i}|} - \frac{1}{|C_{p_i}|-|C_i|} \geq 0,
    \alpha \geq 0
$$
$$
   \Omega_{p_i}  = \left( \frac{|C_{p_i}|+1}{|C_{p_{p_i}}|} - \frac{1}{|C_{p_{p_i}}|-|C_{p_i}|} \right) \cdot (1-\alpha)\geq 0.
$$
which implies
$ b_i = \omega_i \mathcal{B} \geq 0 $.
\end{proof}

\begin{lemma}\label{lem:ic}
In the PRST, agents have no incentives to block any of her children's participation.
\end{lemma}

\begin{proof}
Since the reward $\mathcal{B}$ is fixed, then we only need to consider the coefficients $\omega_i$. 
Suppose an agent $i$'s coefficient becomes $\omega_i'$ when she blocks some of her children from participation. 
Denote $|C_i| - |C_i'|$ by $\Delta|C_i|$. Then for all agent $i$'s ancestor $j\in N$, denote $|C_{p_j}|-|C_j|$ by $|C_{-j}|$, which will not change and we have $|C_j| - |C_j'| = \Delta|C_j|$. Hence
\begin{align*}
    \Omega'_{j} & = \Omega'_{p_j} \left( \frac{|C'_j|+1}{|C'_{p_j}|} - \frac{1}{|C'_{-j}|} \right) ( 1 - \alpha) \\
    & = \Omega'_{p_j} \left( \frac{|C_j|+1-\Delta|C_i|}{|C_{p_j}|-\Delta|C_i|} - \frac{1}{|C_{-j}|} \right) ( 1 - \alpha) \\
    & \leq \Omega'_{p_j} \left( \frac{|C_j|+1}{|C_{p_j}|} - \frac{1}{|C_{-j}|} \right) ( 1 - \alpha)
\end{align*}
Assume $\Omega'_k \leq \Omega_k$ for agent $k$.
Then for any child $m$ of $k$,
    \begin{align}
        \Omega'_{m} &\leq \Omega'_{p_m} \left( \frac{|C_m|+1}{|C_{p_m}|} - \frac{1}{|C_{-m}|} \right) ( 1 - \alpha) \notag\\
        &=\Omega'_{k} \left( \frac{|C_m|+1}{|C_{k}|} - \frac{1}{|C_{-m}|} \right) ( 1 - \alpha)\notag\\
        &\leq \Omega_{k} \left( \frac{|C_m|+1}{|C_{k}|} - \frac{1}{|C_{-m}|} \right) ( 1 - \alpha)\notag\\
        &=\Omega_{m} \label{1}
    \end{align}
The base case is that 
\begin{align}
    \Omega'_{s} = \Omega_{s} = 1
\end{align}
Combining (1) and (2), by induction, we know that for any agent $j$,
\begin{align*}
    \Omega'_{j} \leq \Omega_{j}
\end{align*}
Then, for agent $i$,
\begin{align*}
    \omega'_i
    & = \Omega'_{p_i} \left[ \frac{1}{|C'_{-i}|} + \left( \frac{|C'_i|+1}{|C'_{p_i}|} - \frac{1}{|C'_{-i}|} \right) \alpha \right] \\
    & \leq \Omega_{p_i} \left[ \frac{1}{|C_{-i}|} + \left( \frac{|C_j|+1-\Delta|C_i|}{|C_{p_i}|-\Delta|C_i|} - \frac{1}{|C_{-i}|} \right) \alpha \right] \\
    & \leq \Omega_{p_i} \left[ \frac{1}{|C_{-i}|} + \left( \frac{|C_j|+1}{|C_{p_i}|} - \frac{1}{|C_{-i}|} \right) \alpha \right] = \omega_i
\end{align*}
Therefore, agent $i$ will suffer a loss if the amount of her descendants decreases so she has no incentives to block any of her children's participation.
\end{proof}

\begin{lemma}\label{lem:wbb}
In the PRST, the total share distributed to all agents is exactly  $\mathcal{B}$.
\end{lemma}

\begin{proof}
Considering an agent $i\in N$, we have 
$$ \omega_i = \Omega_{p_i} \left( \frac{1}{|C_{p_i}|-|C_i|} + \left( \frac{|C_i|+1}{|C_{p_i}|} - \frac{1}{|C_{p_i}|-|C_i|} \right) \cdot \alpha \right) $$
and the descendants of $i$ can at most be distributed
$$ \Omega_i = \Omega_{p_i} \left( \frac{|C_i|+1}{|C_{p_i}|} - \frac{1}{|C_{p_i}|-|C_i|} \right) \cdot(1-\alpha) $$ 
in total, which means that
$$ \omega_i+\sum_{j \in C_i} \omega_j \leq \omega_i + \Omega_i = \frac{|C_i|+1}{|C_{p_i}|} \Omega_{p_i} $$ 
where the equation holds if and only if 
    $$\sum_{j \ \in \ C_i}\omega_j = \Omega_{p_i}\left( \frac{|C_i|+1}{|C_{p_i}|} - \frac{1}{|C_{p_i}|-|C_i|} \right) \cdot(1-\alpha)$$
Then consider an agent $k$ whose children are all leaf nodes (i.e., $C_l = \emptyset$ for all $l \in C_k$). We can get
$$ \sum_{l \ \in \ C_k}\omega_l = \Omega_{p_k} \left( \frac{|C_k|+1}{|C_{p_k}|} - \frac{1}{|C_{p_k}|-|C_k|} \right) (1-\alpha) $$
By the definition of $\omega_k$ $$ \omega_k = \Omega_{p_k} \left( \frac{1}{|C_{p_k}|-|C_k|} + \left( \frac{|C_k|+1}{|C_{p_k}|} - \frac{1}{|C_{p_k}|-|C_k|} \right) \cdot \alpha \right) $$
Then $$ \omega_k+\sum_{l \ \in \ C_k}\omega_l = \frac{|C_k|+1}{|C_{p_k}|}\Omega_{p_k} $$

Hence, by induction, $\omega_i+\sum_{j \in C_i} \omega_j = \frac{|C_i|+1}{|C_{p_i}|} \Omega_{p_i}$ holds for all agent $i\in N$. Then the total share of all agents are
\begin{align*}
    \sum_{i\in N} b_i & = \sum_{i\in N} \omega_i \mathcal{B}  = \mathcal{B}\sum_{i\in N} \omega_i \\
    & = \mathcal{B} \sum_{i\in r_s} \left( \omega_i+\sum_{j \ \in \ C_i}\omega_j \right)  \\
    & = \mathcal{B} \cdot \Omega_s \sum_{i\in r_s} \frac{|C_i|+1}{|C_{s}|}  = \mathcal{B}
\end{align*}

Therefore, the total share distributed is exactly $\mathcal{B}$.
\end{proof}

\subsection{Network-based Redistribution Framework}
\par Now we propose our redistribution framework for diffusion auctions. We first introduce a key concept called the \emph{diffusion critical tree}, which reveals who takes the most important role for each agent's participation and simplifies the graph structure.
\par Given a report profile $\theta'$, and the induced graph $G(\theta')$, we can generate a diffusion critical tree $T(\theta')$ from $G(\theta')$ as the following.
\begin{itemize}
    \item $T(\theta')$ is a rooted tree, where sponsor $s$ is the root and there is an edge $(s, i)$ for all $i\in r_s$.
    \item For all agents $i,j \in N$, there is an edge $(i,j)$ if and only if (1) $i$ is a cut-point to disconnect $j$ from $s$, (2) there is no cut-point to disconnect $j$ from $i$.
    Intuitively, agent $i$ is the closest agent to $j$ whose leaving will block $j$'s participation.
\end{itemize}

\begin{figure}[htbp]
    \centering  
	\begin{subfigure}[b]{0.2\textwidth}
        \includegraphics[width=\textwidth]{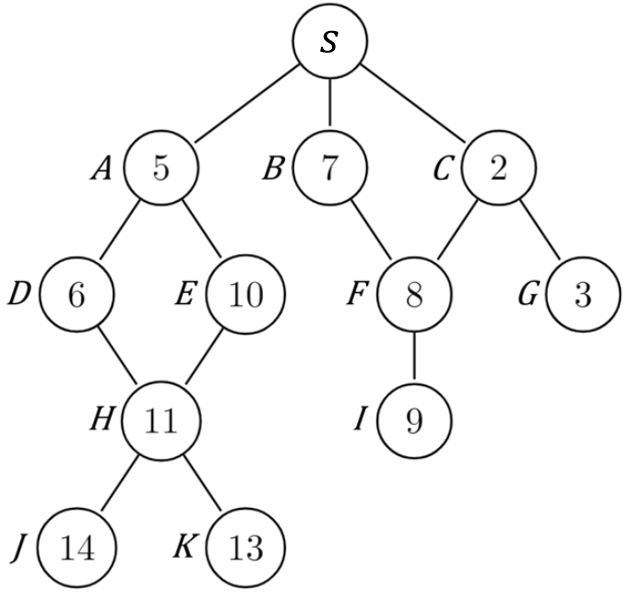}
        \caption*{(1)}
        \label{fig:2.21}
    \end{subfigure}
    	\begin{subfigure}[b]{0.035\textwidth}
        \includegraphics[width=\textwidth]{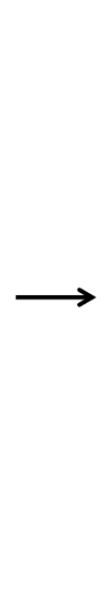}
        \caption*{}
        \label{fig:2.22}
    \end{subfigure}
    \begin{subfigure}[b]{0.23\textwidth}
        \includegraphics[width=\textwidth]{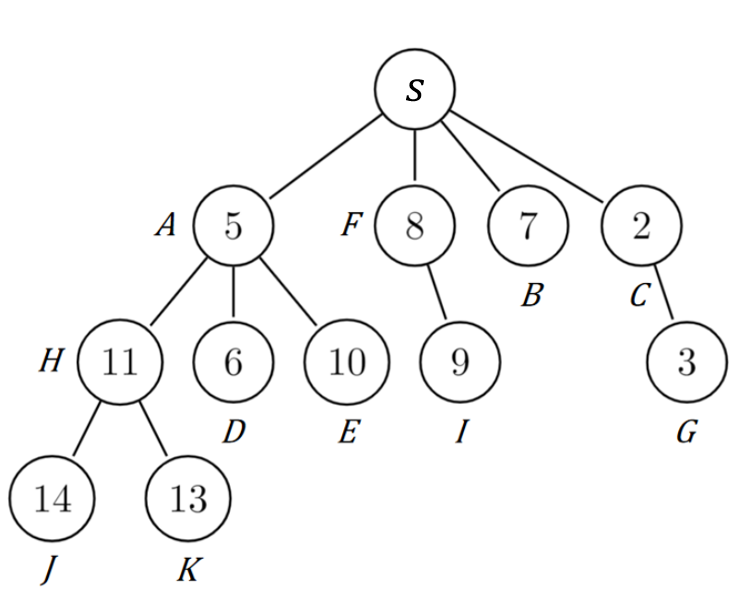}
        \caption*{(2)}
        \label{fig:2.23}
        \end{subfigure}
    \caption{(1) is a network $G(\theta')$. (2) is the corresponding diffusion critical tree $T(\theta')$. The number in each node is the reported valuation of the agent.}
    \label{fig:trans}
\end{figure}

\par Figure~\ref{fig:trans} shows an example of the above process. In Figure~\ref{fig:trans}(1), taking agent $H$ as an example, we can see that the closest agent whose leaving will block her participation is agent $A$. Hence, the parent of agent $H$ in $T(\theta')$ is agent $A$. We call agent $A$ as $H$'s \emph{critical parent}, and agent $H$ and $A$ are $J$'s \emph{critical ancestors}. 
\par Then, we present our network-based redistribution mechanism framework with diffusion auction mechanism $\mathcal{M}^a$ as the input. 

\begin{framed}

 \noindent\textbf{Network-based Redistribution Mechanism Framework (NRMF)}
 
 \noindent\rule{\textwidth}{0.5pt}
 
 \noindent\textsc{Input}: a diffusion auction mechanism $\mathcal{M}^a$ and 
 
$\hspace{1.608em} $ a report profile $\theta'$.
 
 \noindent\rule{\textwidth}{0.5pt}
 
 \begin{enumerate}
     \item Run diffusion auction mechanism $\mathcal{M}^a$ on $\theta'$ and it outputs $\pi^a(\theta')$, $x^a(\theta')$. The revenue gain by $\mathcal{M}^a$ is $S^a(\theta')$. \item Let $\pi(\theta') = \pi^a(\theta')$.
     \item Let $T(\theta')$ be the diffusion critical tree of $G(\theta')$. 
     Let $T_i$ be the subtree in $T(\theta')$ rooted by $i \in N$.
     \item Set $\mathcal{B}= 1$ and run the PRST on $T(\theta')$, and get $b_i$ for each $i\in N$.
     \item Let $\hat{r}_s$ be the set of the neighbours of the sponsor $s$ in $T(\theta')$. And agents in the set are labelled as $(m_1,..,m_{|\hat{r}_s|})$.
     \item For each subtree $T_{m_k}$, $1\leq k\leq |\hat{r}_s|$, set $\theta''$ as
     $$\theta_i'' = \left\{
    \begin{array}{ccl}
    \theta_i'& &{\text{if agent} \ i \notin T_{m_k}} \\
     & &\\
    (0, \emptyset) & &{\text{if agent} \ i \in T_{m_k}}
    \end{array}\right.$$
    and simulate $\mathcal{M}^a$ on $\theta''$. The revenue gained by $\mathcal{M}^a$ is $S^a(\theta'')$. Let $B_k = S^a(\theta'')$.
    \item For each $i\in N$, set share $\hat{b}_i= b_i \cdot B_k$ if $i\in T_{m_k}$ ($1\leq k\leq |\hat{r}_s|$).
    \item Let $R_i(\theta') = \hat{b}_i = \omega_i B_k$  (since  $\mathcal{B} = 1$) and $x_i(\theta') = x^a_i(\theta') - R_i(\theta')$.
 \end{enumerate}
 
 \noindent\rule{\textwidth}{0.5pt}
 
 \noindent\textsc{Output}: the allocation $\pi(\theta')$ and the payment $x(\theta')$.
\end{framed}

\par Intuitively, the diffusion critical tree $T(\theta')$ is divided into $|\hat{r}_s|$ subtrees. The $B_k$ for a subtree $T_{m_k}$ is the revenue obtained by running $\mathcal{M}^a$ on $\theta''$ where $T_{m_k}$ is blocked. The $B_k$ is shared by all agents in $T_{m_k}$, and it is independent of these agents' report profiles. We will show that the procedure of the PRST will not affect the independence of $B_k$ and that the whole mechanism is incentive compatible. It should be noted that as $\alpha$ grows larger, it will redistribute more to the inviters; otherwise, the invitees will receive more. We can flexibly change the value of $\alpha$ in the PRST according to the practical requirements without affecting the properties. Note that when the input mechanism is simply running traditional mechanism (e.g., the VCG mechanism) in traditional settings (i.e., all agents are sponsor's neighbours),  NRMF still works. Therefore, our framework is a general solution for redistribution problems with or without networks.

\par We now demonstrate that NRMF can satisfy the desirable properties of IR, IC, and non-deficit if the input diffusion auction $\mathcal{M}^a$ is IR, IC, non-deficit and revenue monotonic. 

\begin{theorem}
The instance of diffusion redistribution mechanism given by NRMF is \textbf{individually rational} (IR) if the input diffusion auction $\mathcal{M}^a$ is IR and non-deficit.
\end{theorem}\label{tem:ir}
\begin{proof}
Consider agent $i$'s utility, $i \in N$.
$$u_i(\theta_i,\theta') = \pi^a_i(\theta')v_i - x^a_i(\theta') + R_i(\theta')$$

The part of $\pi^a_i(\theta')v_i - x^a_i(\theta')$ is the utility of $i$ in the diffusion auction mechanism $\mathcal{M}^a$, which is non-negative if agent $i$ reports her type truthfully since $\mathcal{M}^a$ is IR, i.e.,
\begin{equation}
    \pi^a_i((\theta_i, \theta_{-i}'))v_i - x^a_i((\theta_i, \theta_{-i}')) \geq 0
\end{equation}

On the other hand, considering the part $R_i(\theta') = \omega_i B_k$, for $i\in T_{m_k}$, $1\leq k\leq |\hat{r}_s|$, since $\mathcal{M}^a$ is non-deficit, then
\begin{equation}
    B_k \geq 0
\end{equation}
Finally, according to Lemma~\ref{lem:ir}, we have
\begin{equation}
    \omega_i \geq 0
\end{equation}
Combining (3), (4), and (5), we get
$$
    u_i(\theta_i,(\theta_i, \theta_{-i}')) \geq 0
$$
Therefore, the mechanism is individually rational.
\end{proof}

\begin{theorem}
The instance of diffusion redistribution mechanism given by NRMF is \textbf{incentive compatible} (IC) if the input diffusion auction $\mathcal{M}^a$ is IC.
\end{theorem}
\begin{proof}
Consider agent $i$'s utility when she truthfully reports her type $\theta_i$, $i \in N$.
$$
    u_i(\theta_i,\theta') = \pi^a_i(\theta')v_i - x^a_i(\theta') + R_i(\theta')
$$
where $\theta' = (\theta_i, \theta_{-i}')$.
\par If agent $i$ misreports $\theta'_i$, then suppose $\ell$ is the loss of $i$ in diffusion auction mechanism $\mathcal{M}^a$ because of misreport. Since $\mathcal{M}^a$ is IC, $\ell \geq 0$ and we have:
$$
    u_i(\theta_i, \theta')- u_i(\theta_i, \theta'') = \ell + R_i(\theta')-R_i(\theta'')
$$
where $\theta'' = (\theta_i', \theta_{-i}')$,
which means agent $i$ reports profile $\theta_i'$ that may be different from the type $\theta_i$, and the report profile of all agents except for $i$ remains $\theta_{-i}'$.

\par Suppose $i\in T_{m_k}$ in $T(\theta')$. Since $i$ cannot change her position in critical diffusion tree, which only depends on agents who invite her, then $i\in T_{m_k}$ in $T(\theta'')$, too. Let $R_i(\theta') = \omega_i B_k$ and $R_i(\theta'') = \omega_i' B_k'$. Since $B_k$ and $B_k'$ is the revenue $\mathcal{M}^a$ can achieve without the participation of agents in $T_{m_k}$, we have $B_k = B_k'$. Finally, with misreporting, $i$ can only decrease the number of her descendants in $T_{m_k}$. According to Lemma~\ref{lem:ic}, we have $\omega_i \geq \omega_i'$. Therefore,
\begin{align*}
    u_i(\theta_i, \theta') - u_i(\theta_i, \theta'') & = \ell + R_i(\theta')-R_i(\theta'') \\
    & = \ell + (\omega_i - \omega'_i) \cdot B_k \geq 0
\end{align*}
from which we can conclude that the mechanism is IC.
\end{proof}

A diffusion auction is revenue monotonic if the revenue of the auction monotonically increases as the number of participants increases.

\begin{definition}
    A diffusion auction is \textbf{revenue monotonic} if for all $\theta', \theta''\in \Theta$ with $D_s(G(\theta')) \subseteq D_s(G(\theta''))$, and for all $i\in D_s(G(\theta'))$, $v_i' = v_i''$ and $r_i'\subseteq r_i''$, we have $S(\theta') \leq S(\theta'')$.
\end{definition}

Furthermore, if all the new participants have relatively small valuations, it should not affect the revenue of the sponsor since they have no contribution. We call it \emph{revenue invariance}.

\begin{definition}
    A diffusion auction is \textbf{revenue invariant} if 
    \begin{itemize}
        \item for all $\theta', \theta''\in \Theta$ with $D_s(G(\theta')) \subseteq D_s(G(\theta''))$, and for all $i\in D_s(G(\theta'))$, $v_i' = v_i''$ and $r_i'\subseteq r_i''$;
        \item for all agents in $D_s(G(\theta''))\setminus D_s(G(\theta'))$, any of them cannot be the winner \textbf{even if} we remove the winner and all her critical ancestors under $\theta'$,
    \end{itemize} then we have $S(\theta') = S(\theta'')$.
\end{definition}\label{def:ri}

It is easy to prove that almost all the existing auction mechanisms (with or without diffusion)~\cite{vickrey1961counterspeculation,DBLP:conf/aaai/LiHZZ17,zhang2020incentivize,li2022diffusion}, satisfy revenue invariance.

\begin{theorem}
The instance of diffusion redistribution mechanism given by NRMF is \textbf{non-deficit} (ND) if the input diffusion auction $\mathcal{M}^a$ is revenue monotonic.
\end{theorem}
\begin{proof}

\par According to Lemma~\ref{lem:wbb}, we have
$ \sum_{i\in N} \omega_i = 1 $.
\par On the other hand, since $\mathcal{M}^a$ is revenue monotonic, then for each subtree $T_{m_k}$ in the diffusion critical tree $T(\theta')$, we have $B_k\leq S^a(\theta')$. Hence, 
\begin{align*}
    S(\theta') & = \sum_{i \in N} x_i(\theta')
    = S^a(\theta') - \sum_{k} \sum_{i\in T_{m_k}} \omega_i \cdot B_k \\
    & \geq S^a(\theta') - \sum_{i\in N} \omega_i \cdot S^a(\theta') = 0
\end{align*}
Therefore, the mechanism is non-deficit.
\end{proof}

Then we will discuss the properties of ABB and $\epsilon$-ABB. 
In the traditional setting, when we talk about ABB, the increase of agents corresponds to the increase of the sponsor's neighbours in our setting. However, on the social networks, it is unreasonable to only increase sponsor's neighbours. So the number of other agents' neighbours on the graph will grow together. Due to the existence of common neighbours, the increase of each agent's neighbours on the origin graph is hard to describe. So we will discuss agents that grow to infinity in the diffusion critical tree, which reflects the invitation relationship in social networks. If all agents have the same probability of inviting someone new in the critical tree, we define it as evenly growing.

\begin{definition}
 A diffusion critical tree $T$ is \textbf{evenly growing} if for all subtree $T_i \subset T $ where $i \in N$, we have
    $$\lim_{n \to \infty} \frac{|T_i|}{n} = 0. $$
\end{definition}

Since the process of the diffusion auction is naturally seeking more agents, then we are also interested in the increase of the agents as a continuous process, i.e., the sponsor's neighbours are fixed and the critical tree only grows in height. 
If assuming each neighbour of the sponsor have the same potential in terms of the number of agents in their leading branches, we call the critical tree is branch-independent growing.

\begin{definition}
    A diffusion critical tree $T$ is  \textbf{branch-independent growing} if for all subtree $T_i \subset T $ rooted by $i \in \hat{r}_s$, we have
    $$\lim_{n \to \infty} \frac{|T_i|}{n} = \frac{1}{|\hat{r}_s|}. $$
\end{definition}




\begin{lemma}\label{lem:ir}
If a diffusion auction $\mathcal{M}^a$ is IR and ND, we have $0\leq S^a(\theta')\leq \overline{v}$, where $\overline{v}$ is the upper bound of all possible valuations\footnote{Otherwise, the valuation can be infinity, which is not reasonable in practice.}. 
\end{lemma}
\begin{proof}
    Since the the diffusion auction $\mathcal{M}^a$ is ND, we can get:
    $$ S^a(\theta') \geq 0$$
    
    We have mentioned in Theorem~\ref{tem:ir} that $\pi^a_i(\theta')v_i - x^a_i(\theta')$ is the utility of agent $i$ in the diffusion auction mechanism $\mathcal{M}^a$ and it is non-negative when $\mathcal{M}^a$ is IR. So for agent $w$ who wins the item, her $\pi^a_w(\theta') =1$ and payment $x^a_w(\theta') \leq v_w$. For others, they will not get the item and their payment $x^a_i(\theta') \leq 0$. Therefore, 
    $$S^a(\theta') = \sum_{i \in N} x^a_i(\theta') \leq v_w$$
    
    The report valuations of all agents have a finite upper bound $\overline{v}$ ($\mathop{max}\limits_{i\in N}$ $v_i \leq \overline{v}$). 
    Hence, $0\leq S^a(\theta')\leq \overline{v}$.

\end{proof}

\begin{theorem}\label{tem:abb}
If the input diffusion auction mechanism $\mathcal{M}^a$ is IR, non-deficit and revenue invariant, then the instance of NRMF is
\begin{itemize}
    \item[1.] \textbf{asymptotically budget-balanced} (ABB) when the diffusion critical tree $T(\theta')$ is evenly growing;

    \item[2.] \textbf{$\epsilon$-asymptotically budget-balanced} ($\epsilon$-ABB) when the diffusion critical tree $T(\theta')$ is branch-independent growing.
\end{itemize}
\end{theorem}
\begin{proof}
After we first run diffusion auction mechanism $\mathcal{M}^a$ on $\theta'$, $T_{m_w}$ is the subtree that contains the item winner and we can get $S^a(\theta')$. Then when we remove the agents in $T_{m_w}$ from the corresponding $G(\theta')$ and run $\mathcal{M}^a$ again, the new winner is in subtree $T_{m_{w'}}$. Agent $i \in N\setminus(T_{m_w} \cup T_{m_{w'}})$ will never be the item winner even if we remove the agents in $T_{m_w}$ or $T_{m_{w'}}$. Therefore, according to the definition of revenue invariance, the attendance of these agents does not influence the revenue gained by the $\mathcal{M}^a$. Hence, we can get $B_k = S^a(\theta')$ for $k\in\{1,\dots,{|\hat{r}_s|}\}\setminus\{m_w,m_{w'}\}$. Let $B_{m_w} = S^a(\theta''_{-m_w})$ and $B_{m_w'} = S^a(\theta''_{-m_{w'}})$, and the remaining part of revenue that has not been redistributed is
\begin{align*}
 S(\theta')  &= S^a(\theta') - \sum_{k=1}^{|\hat{r}_s|} \sum_{i\in T_{m_k}} \omega_i\cdot B_k \\ 
&= S^a(\theta') - \sum_{k=1}^{|\hat{r}_s|} \frac{|T_{m_k}|}{n} \cdot B_k 
\\  &=  \left(  1-\sum_{k\in\{1,\dots,{|\hat{r}_s|}\}\setminus\{m_w,m_{w'}\}} \frac{|T_{m_k}|}{n} \right)S^a(\theta') \\  &\quad -  \frac{|T_{m_w}|}{n} \cdot S^a(\theta''_{-m_w}) - \frac{|T_{m_{w'}}|}{n} \cdot S^a(\theta''_{-m_{w'}})  \\
    & = \frac{|T_{m_w}|}{n} \cdot (S^a(\theta') - S^a(\theta''_{-m_w})) \\  &\quad + \frac{|T_{m_{w'}}|}{n} \cdot (S^a(\theta') - S^a(\theta''_{-m_{w'}}))
\end{align*}
According to the Lemma~\ref{lem:ir}, $S^a(\theta')$, $S^a(\theta''_{-m_w})$ and $S^a(\theta''_{-m_{w'}})$ are bounded. Obviously, $(S^a(\theta') - S^a(\theta''_{-m_w}))$ and $(S^a(\theta') - S^a(\theta''_{-m_{w'}}))$ are also bounded. 
\begin{itemize}

    \item [1.] When the diffusion critical tree $T(\theta')$ is evenly growing,  $\frac{|T_{m_w}|}{n}$ and $\frac{|T_{m_{w'}}|}{n}$ approach to 0 if $n$ approaches to infinity. Hence,
$$\lim_{n \to \infty} S(\theta') = 0$$
Therefore, the instance is asymptotically budget-balanced.


  \item [2.] When the diffusion critical tree $T(\theta')$ is branch-independent growing, let $(S^a(\theta') - S^a(\theta''_{-m_w}))$ and $(S^a(\theta') - S^a(\theta''_{-m_{w'}})) $ less than $\overline{S^a}$. We can get
  $$\lim_{n \to \infty} S(\theta') \leq \frac{2}{|\hat{r}_s|} \cdot \overline{S^a} = \epsilon$$ where $\epsilon$ is a constant. Therefore, the instance is $\epsilon$-asymptotically budget-balanced.
\end{itemize}
\end{proof}

Note that if the corresponding diffusion auction is also revenue monotonic, $S^a(\theta') \geq S^a(\theta''_{-m_w})$ and $S^a(\theta') \geq S^a(\theta''_{-m_{w'}})$. Then the $\epsilon$ in Theorem~\ref{tem:abb} can be $\frac{2}{|\hat{r}_s|} \cdot \overline{v}$.

\section{Instances of The Redistribution Mechanism Framework}

In our network-based redistribution mechanism framework, if we require the output mechanism to be IC and IR, then the input diffusion auction mechanism should also be IC, IR and non-deficit. The largest known set of diffusion auction mechanisms with the above properties is Critical Diffusion Mechanism (CDM)~\cite{DBLP:conf/ijcai/LiHZY19}. Especially, the first diffusion auction mechanism, Incentive Diffusion Mechanism (IDM)~\cite{DBLP:conf/aaai/LiHZZ17} is also a member in CDM, which has the highest efficiency. In this section, we input IDM and another mechanism in CDM called Threshold Neighbourhood Mechanism (TNM)~\cite{li2022diffusion} into our framework to see the outcomes. 


For convenience, we briefly introduce the idea of the IDM and TNM with our notations. Both IDM and TNM first find the agent with the highest valuation and their critical ancestors. Then the mechanisms check these agents from the sponsor to the agent with the highest valuation. For IDM, each agent will pay a certain amount to her critical parent, which is the highest valuation after removing herself from the graph. Then, she will acquire the item temporarily. If her valuation is the highest after removing her critical descendants, she will be the winner and the mechanism terminates. By contrast, agent under TNM will remove all her descendants (including non-critical ones) when we check the winner.
In addition, when agents get item, 
they will also pay the same amount of money as the IDM, but their critical parents will just get the value of the highest valuation when we ignore all their descendants. The rest of the payments will be directly given to the sponsor.

\begin{figure*}[htbp]
    \centering  
	\begin{subfigure}[b]{0.25\textwidth}
        \includegraphics[width=\textwidth]{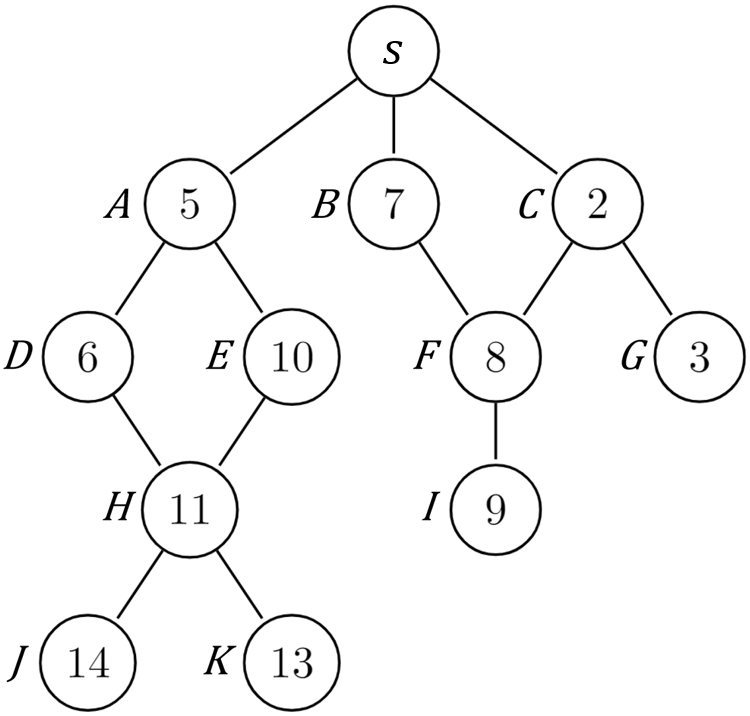}
        \caption*{(1)}
        \label{fig:5.1}
    \end{subfigure}
    \begin{subfigure}[b]{0.075\textwidth}
        \includegraphics[width=\textwidth]{white.PNG}
        \caption*{}
    \end{subfigure}
	\begin{subfigure}[b]{0.25\textwidth}
        \includegraphics[width=\textwidth]{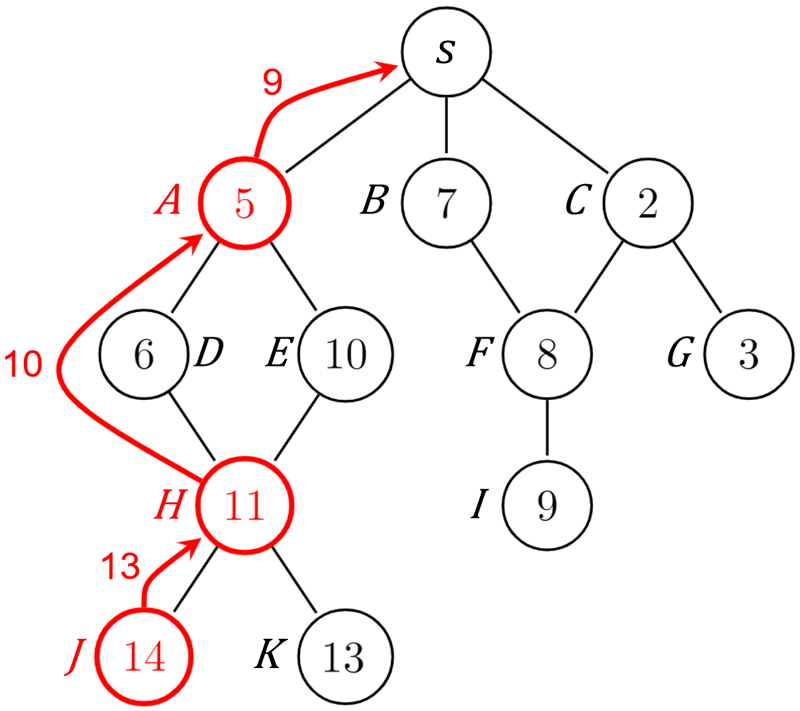}
        \caption*{(2)}
        \label{fig:5.1}
    \end{subfigure}
        \begin{subfigure}[b]{0.075\textwidth}
        \includegraphics[width=\textwidth]{white.PNG}
        \caption*{}
    \end{subfigure}
    	\begin{subfigure}[b]{0.25\textwidth}
        \includegraphics[width=\textwidth]{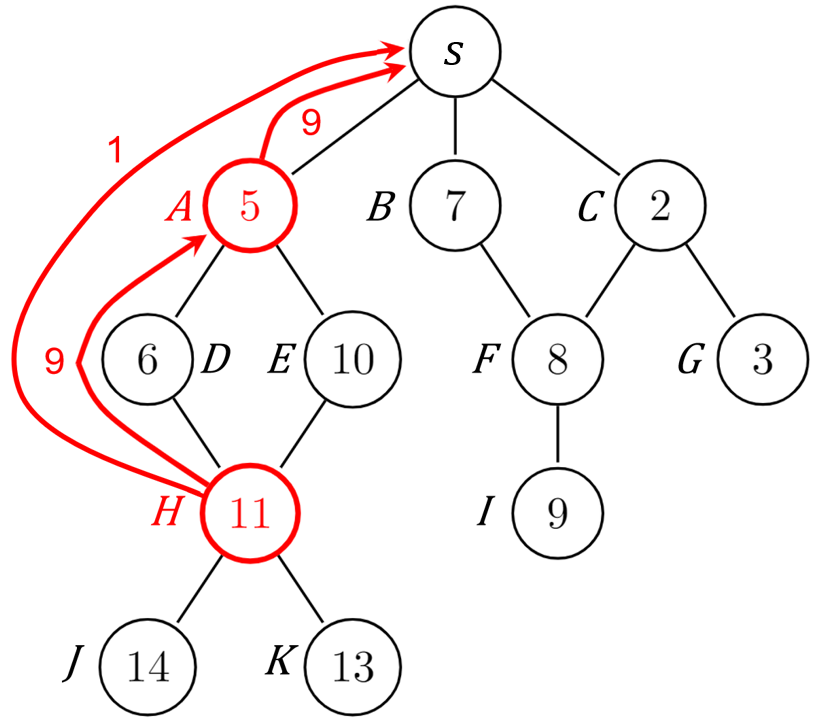}
        \caption*{(3)}
        \label{fig:5.1}
    \end{subfigure}
    \\
    \begin{subfigure}[b]{0.33\textwidth}
        \includegraphics[width=\textwidth]{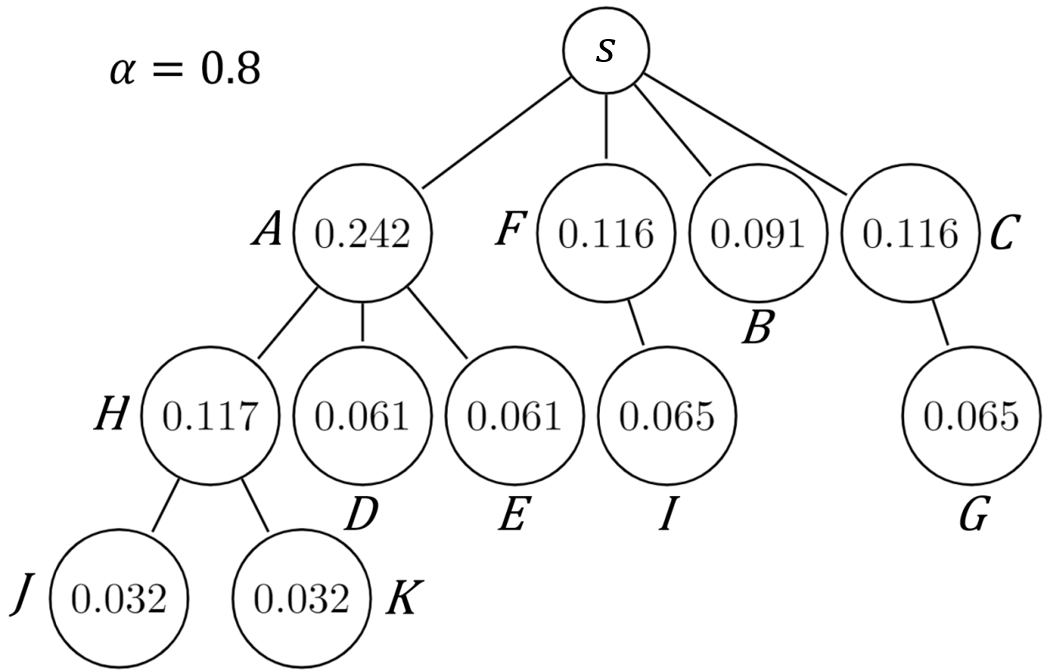}
        \caption*{(4)}
        \label{fig:5.1}
    \end{subfigure}
    	\begin{subfigure}[b]{0.33\textwidth}
        \includegraphics[width=\textwidth]{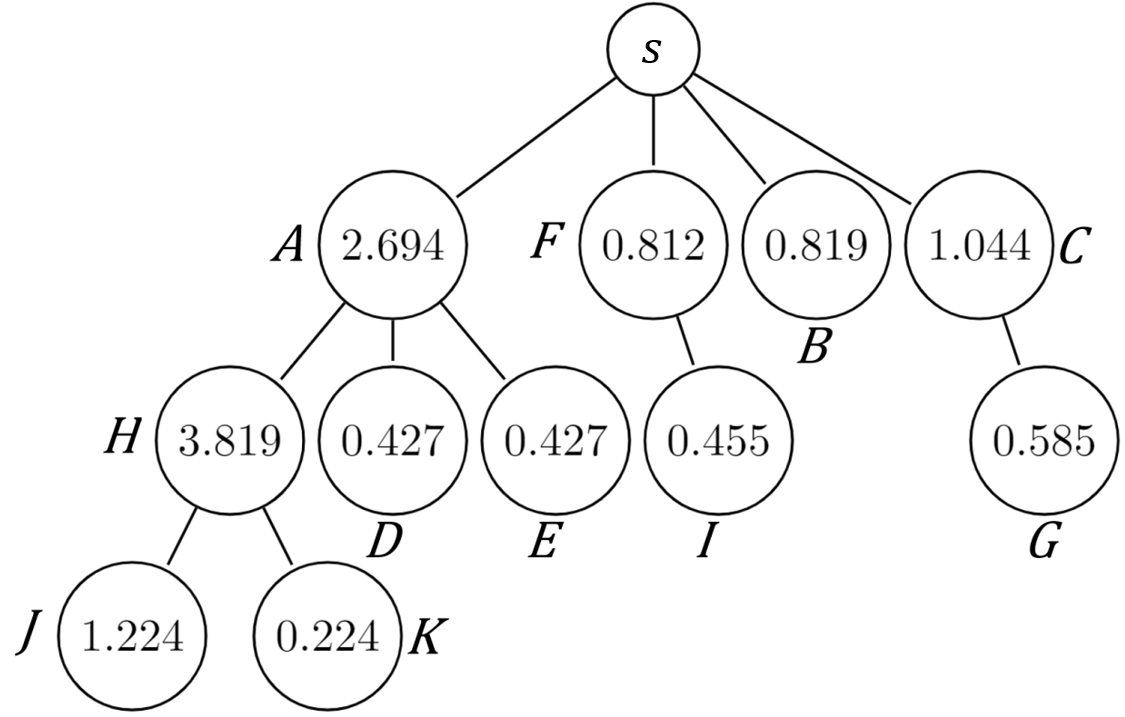}
        \caption*{(5)}
        \label{fig:5.1}
    \end{subfigure}
    	\begin{subfigure}[b]{0.33\textwidth}
        \includegraphics[width=\textwidth]{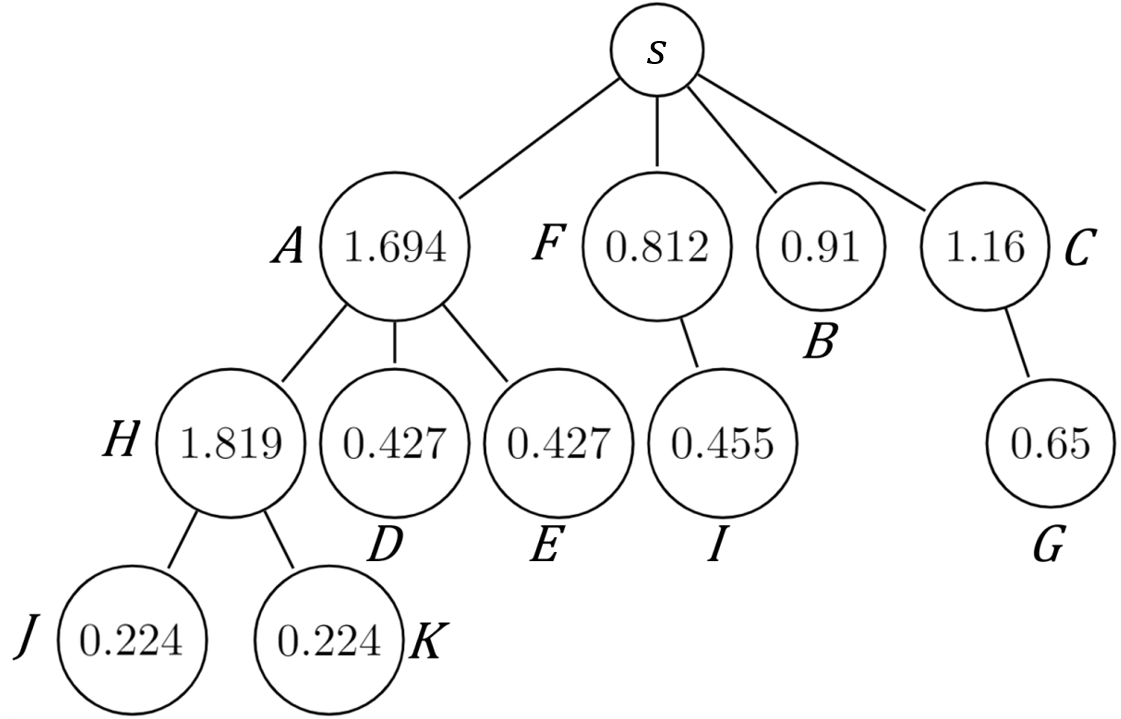}
        \caption*{(6)}
        \label{fig:5.1}
    \end{subfigure}

    \caption{Two running examples of the IDM based and TNM based instance mechanisms of NRMF. (1) shows a graph $G(\theta')$ generated from a report profile $\theta'$. (2) is the corresponding graph $G(\theta')$ under the diffusion action mechanism IDM, where the red arrows mean the payment transfers of agents. The item is allocated to agent $J$. (3) is the corresponding graph $G(\theta')$ under TNM. The red arrows also mean the payment transfers of agents. However, the item is allocated to agent $H$ which is the difference. (4) shows the coefficients $\omega_i$ of agents and we set $\alpha = 0.8$. (5) and (6) shows the final utilities corresponding to IDM mechanism input and TNM mechanism input respectively.}
    \label{fig:IDM}
\end{figure*} 

We show two running examples of the redistribution mechanism with IDM and TNM in Figure~\ref{fig:IDM}. Both IDM and TNM first find the agent $J$ who has the highest valuation. Then they check the critical path $s \rightarrow A \rightarrow H \rightarrow J$. Under IDM, the item is allocated to the agent $J$ and the sponsor's revenue is $9$. After removing the all agents in the subtree rooted by agent $A$, the revenue gained by the sponsor under IDM is $7$. Therefore, if we set $\alpha = 0.8$, the final utility of agent $J$ is $14-13 + 0.032 \times 7 = 1.224$. When it comes to TNM, the agent $H$ will win the item because she is the agent who reports the highest valuation after removing the agent $J$ and $K$ on the critical path. $H$'s critical parent $A$ will get the highest valuation $9$ after removing all $A$'s descendants and the rest $1$ of $H$'s payment will be directly given to the sponsor. Similarly, after removing the all agents in the subtree rooted by $A$, sponsor still gets $7$. The final utility of agent $H$ is $11-10 + 0.117 \times 7 = 1.819$.

\section{Discussion and Conclusion}
In this paper, we focus on redistribution mechanism design on social networks, where a sponsor wants to incentivize agents to invite their neighbours to participate in, and allocate a single item without seeking any profit. 
To achieve the goal, we propose a network-based redistribution mechanism framework (NRMF) that can construct a diffusion redistribution mechanism from any diffusion auction.
The NRMF will maintain the properties of incentive compatibility and individual rationality of the original diffusion auction, and also be non-deficit if the original diffusion auction is revenue monotonic.

Moreover,  without affecting the efficiency of the diffusion auction, NRMF can achieve the property of asymptotically budget-balanced (ABB) or $\epsilon$-asymptotically budget-balanced. We consider ABB rather than budget-balance (BB), which requires that all revenue be returned back to buyers ideally. It is common even under the traditional settings because no mechanisms can satisfy all the properties of efficiency 
, IC, IR and BB according to Green-Laffont impossibility theorem~\cite{green1979incentives}. 
When it comes to our setting, the above impossibility theorem still holds. The reason is that the special cases where all agents are directly connected to the sponsor in the network are equivalent to the cases in the traditional settings. However, in diffusion auction mechanism design, efficiency is usually abandoned because it is even impossible to design a diffusion auction mechanism that satisfies efficiency, IC, IR and non-deficit simultaneously~\cite{li2022diffusion}.

If we require the output mechanism to be IC and IR in NRMF, then the input mechanism must be IC, IR and non-deficit. The output redistribution mechanism achieves the same level of social welfare as the input mechanism. 

In theory, without efficiency, the impossibility that BB cannot be achieved is missing in our setting. Actually, it is possible to satisfy the properties of IC, IR and BB simultaneously. For example, when we input the fixed pricing diffusion auction mechanisms into our framework, the corresponding output mechanisms are always BB. It is an interesting future work that finding the relationship between the efficiency and the residual budget after the redistribution. 

\begin{acks}
This work is supported by Science and Technology Commission of Shanghai Municipality (No. 23010503000 and No. 22ZR1442200) and Shanghai Frontiers Science Center of Human-centered Artificial Intelligence (ShangHAI).
\end{acks}

\bibliographystyle{ACM-Reference-Format} 
\bibliography{FNRM}


\end{document}